\documentclass[12pt,draftclsnofoot,onecolumn]{IEEEtran}

\usepackage{amsfonts}
\usepackage{times}
\usepackage{latexsym}
\usepackage{amssymb}
\usepackage{amsmath}
\usepackage{cite}
\usepackage{verbatim}

\newcommand{\bydef}{\triangleq}

\def\SNR{{\textsf{SNR}}}

\def\bydef{:=}

\def\bb0{{\mathbb{0}}}

\def\bydef{:=}

\def\bb{{\mathbf{b}}}

\def\bff{{\mathbf{f}}}
\def\bg{{\mathbf{g}}}
\def\bh{{\mathbf{h}}}

\def\bq{{\mathbf{q}}}

\def\bt{{\mathbf{t}}}
\def\bu{{\mathbf{u}}}
\def\bv{{\mathbf{v}}}

\def\b0{{\mathbf{0}}}

\def\bA{{\mathbf{A}}}
\def\bB{{\mathbf{B}}}

\def\bQ{{\mathbf{Q}}}
\def\bR{{\mathbf{R}}}
\def\bS{{\mathbf{S}}}

\def\bbC{{\mathbb{C}}}

\def\bbE{{\mathbb{E}}}

\def\bbN{{\mathbb{N}}}

\def\bydef{:=}

\def\sf0{{\mathsf{0}}}

\usepackage{graphicx}
\usepackage{amssymb}
\usepackage{amsfonts}
\usepackage{amsmath}
\usepackage{latexsym}
\begin{document}

\newtheorem{thm}{Theorem}
\newtheorem{lemma}{Lemma}
\newtheorem{rem}{Remark}
\newtheorem{exm}{Example}
\newtheorem{prop}{Proposition}
\newtheorem{defn}{Definition}
\newtheorem{cor}{Corollary}
\def\proof{\noindent\hspace{0em}{\itshape Proof: }}
\def\endproof{\hspace*{\fill}~\QED\par\endtrivlist\unskip}
\def\bh{{\mathbf{h}}}
\def\SNR{{\mathsf{SNR}}}
\title{Transmission Capacity of Spectrum Sharing Ad-hoc Networks with Multiple Antennas}
\author{
Rahul~Vaze
\thanks{Rahul~Vaze is with the School of Technology and Computer Science, Tata Institute of Fundamental Research, Homi Bhabha Road, Mumbai 400005, vaze@tcs.tifr.res.in. }}

\date{}
\maketitle
\noindent
\begin{abstract}
Two coexisting ad-hoc networks, primary and secondary, are considered, where  each node of the primary 
network has a single antenna, while each node of the secondary network is equipped with multiple antennas. 
Using multiple antennas, each secondary transmitter uses some of its  spatial transmit degrees of freedom (STDOF) to null its interference towards the primary receivers, while each secondary receiver employs interference cancelation using some of its spatial receive degrees of freedom (SRDOF).
This paper derives the optimal STDOF for nulling and SRDOF for interference cancelation that maximize the scaling of the  transmission capacity of the secondary network with respect to the number of antennas, when the secondary network 
operates under an outage constraint at the primary receivers. With a single receive antenna, using a fraction of the total STDOF for nulling at each secondary transmitter maximizes the transmission capacity. With multiple transmit and receive antennas and fixing all but one STDOF for nulling, using a fraction of the total SRDOF to cancel the nearest interferers maximizes the transmission capacity of the secondary network.
\end{abstract}

\section{Introduction}
With ever increasing demand for bandwidth, extensive
research has focussed on the intelligent usage of available spectrum \cite{FCC, Devroye2006, Devroye2006a, Jovicic2009}. One of the key ideas to improve
the spectrum utilization is the coexistence of secondary networks together with the primary/licensed network, e.g.  the
use of cognitive radios. Cognitive radios are intelligent devices which
continuously sense the spectrum and schedule transmissions in unutilized frequency bands under a constraint 
on the interference they cause to the primary receivers. 
Spectrum efficiency  can be further improved by using multiple antennas at  cognitive radios, where multiple antennas are used for 
nulling the interference they cause to the primary receivers
or for transmit beamforming or receive interference cancelation  \cite{Telatar1999, Tarokh1999a,ZhangCog2008,KannanCog2009}. 

In this paper we consider the coexistence of two ad-hoc wireless  networks  (primary/licensed and secondary). 
In an ad-hoc wireless network, multiple transmitter-receiver pairs
communicate simultaneously in an uncoordinated manner without the help of
any fixed infrastructure. The primary ad-hoc network is assumed to be oblivious to the presence of the secondary ad-hoc network, and  the secondary ad-hoc network operates under an outage constraint at the primary receivers.
Each node of the primary network is assumed to have a single antenna, while each node of the secondary network is equipped with multiple antennas. We are interested in answering the question: how does the  transmission capacity of the secondary network scale with multiple antennas at secondary nodes, where the transmission capacity
is the maximum allowable intensity of nodes, satisfying a per transmitter receiver rate, and outage probability constraint \cite{Weber2005, Weber2007, Weber2008, Baccelli2006}.

In prior work, the throughput scaling of secondary networks with respect to the number of secondary nodes under an outage constraint at the primary receivers has been studied in \cite{VuCog2009, VuCog2009a, CuiCog2010}.
With a single transmit and receive antenna at the secondary nodes, upper and lower bounds on the transmission capacity of the secondary network have been derived in \cite{HuangCog2009}, 
while an exact transmission capacity expression of the secondary network has been derived in \cite{CuiCog2009} when the  path-loss exponent is four. For a single secondary transmitter-receiver pair, opportunistic spectrum sharing using multiple antennas has been proposed and analyzed in \cite{ZhangCog2008,KannanCog2009}.
To the best of knowledge, however, no work has been reported on the scaling of the  transmission capacity of the secondary networks with respect to the number of antennas available at secondary nodes.

In this paper we assume that each secondary transmitter has $N$ antennas, while each secondary receiver has $M$ antennas.
Each secondary transmitter is assumed to send a single data stream through its multiple antennas.
Multiple  antennas at each secondary transmitter are used for partial nulling, where some spatial transmit degrees of freedom (STDOF) are used for nulling its interference towards the primary receivers, and the rest of the STDOF are used for beamforming towards its corresponding secondary receiver. Similarly, multiple  antennas at each secondary receiver are used for partial interference cancelation, where some spatial receive degrees of freedom (SRDOF) are used for canceling the interference from both the primary and secondary transmitters, and the rest SRDOF are used to  increase the strength of the signal of interest. Our results are summarized as follows. 
\begin{itemize}
\item {\bf $N=M=1$:} We derive an exact expression for the secondary transmission capacity. We characterize the increase in the transmission capacity of the secondary network with respect to the  increase in the allowed outage probability tolerance at the primary receivers.

\item {\bf Arbitrary $N$, $M=1$:} 
Using a fraction of the total STDOF at each secondary transmitter for nulling its interference towards the nearest primary receivers, and the rest of the STDOF for transmit beamforming maximizes the upper and lower bound on the secondary transmission capacity. The secondary transmission capacity lower bound scales as  $ \min\{N^{\frac{2}{\alpha}}, N^{1-\frac{2}{\alpha}} \}$, and the upper bound 
scales as $N^{\frac{2}{\alpha}}$, where $\alpha$ is the path-loss exponent.

\item {\bf $N=1$, Arbitrary $M$:} 
The transmission capacity is independent of $M$. 

\item {\bf Arbitrary $N$ and $M$ ($M\ge N$):}
With $N-1$ STDOF for interference nulling at each secondary transmitter,  using a fraction of the total SRDOF at each secondary receiver for 
canceling the nearest interferers at each secondary receiver,  maximizes the upper and lower bound on the secondary transmission capacity.
With $M=N$, the secondary transmission capacity lower bound scales as  $  N^{1-\frac{2}{\alpha}}$, and the upper bound 
scales linearly in $N$.

\end{itemize}

Our results show that the transmission capacity of the secondary network scales sublinearly with the number of transmit antennas with or without multiple  antennas at each secondary receiver.   
In comparison, when no secondary nodes are present, the transmission capacity of the primary network  scales 
linearly with the number of receive antennas even with a single transmit antenna \cite{Jindal2008a, 
Vaze2009}. The  transmission capacity with coexisting networks is reduced because the secondary network is required to operate under two outage constraints: one at the
primary receivers, and the other at the secondary receivers. 

Another interesting thing to note is the role of transmit antennas in the secondary network. Our results show that 
with a single transmit antenna, the transmission capacity of the secondary network does not scale with the number of receive antennas. 
This is in contrast to the result of \cite{Jindal2008a, Vaze2009}, where without the secondary network, the transmission capacity with a single transmit antenna is shown to scale linearly with the number of receive antennas. 
Our result can be explained by noting that with only a single transmit antenna, none of 
secondary transmitters can null their interference towards any of the primary receivers, and hence the transmission capacity is bottlenecked by the outage constraint at the
primary receivers. Consequently, the transmission capacity of the secondary network is independent of the number of  antennas at the secondary  receivers. 
Thus, to maximize the transmission capacity of the secondary network, the number of secondary transmit antennas should be similar  to the number of secondary receive antennas.

{\it Notation:}
Let ${\bA}$ denote a matrix, ${\bf a}$ a vector and
$a_i$ the $i^{th}$ element of ${\bf a}$. Transpose and conjugate transpose is denoted by $^T$, and $^*$, respectively. 
The expectation of function $f(x)$ with respect to $x$ is denoted by
${\bbE}(f(x))$.
A circularly symmetric complex Gaussian random
variable $x$ with zero mean and variance $\sigma^2$ is denoted as $x
\sim {\cal CN}(0,\sigma^2)$. Let $S_1$ be a set and $S_2$ be a subset of $S_1$. Then $S_2 \backslash S_1$ denotes the set of elements of $S_1$ that do not belong to $S_2$. Let $f(n)$ and $g(n)$ be two function defined on some subset of real numbers.
Then we write $f(n) = \Omega(g(n))$ if
$\exists \ k > 0, \ n_0, \ \forall \ n>n_0$, $|g(n)| k \le |f(n)|$,
$f(n) = {\cal O}(g(n))$ if $\exists \ k > 0, \ n_0, \ \forall \ n>n_0$, $|f(n)| \le |g(n)| k$, and
$f(n) = \Theta(g(n))$ if $\exists \ k_1, \ k_2 > 0, \ n_0, \ \forall \ n>n_0$,
$|g(n)| k_1 \le |f(n)| \le |g(n)| k_2$. We use the symbol
$\bydef$  to define a variable.

\section{System Model}
\label{sec:sys} Consider an ad-hoc network with two sets of nodes: primary and secondary. 
Each primary and secondary transmitter has a primary and secondary receiver associated with it, located at distance $d_p$ and $d_s$ in random direction, respectively. The primary nodes are oblivious to the presence of secondary nodes. The secondary nodes (both transmitters and receivers) are aware of the primary nodes, and try to maximize the transmission capacity \cite{Weber2005} of the secondary network, subject to a constraint on the added outage probability they cause at any primary receiver. 
The locations of primary and secondary transmitters are modeled as two independent homogenous Poisson point
processes (PPPs) on a two-dimensional plane with intensity $\lambda_1$, and $\lambda_2$, respectively.
We consider a slotted ALOHA like random access protocol, where each transmitter attempts to transmit with an access probability $P_a$, independently of all other transmitters. Consequently, the active primary and secondary transmitter processes are also homogenous
PPPs on a two-dimensional plane with intensity
$\lambda_p = P_a\lambda_1$, and $\lambda_s = P_a\lambda_2$. 
Let the location of the $n^{th}$ active
primary transmitter be $T_{pn}$, and the $n^{th}$ active secondary transmitter be $T_{sn}$. The set of all active primary and secondary 
transmitters is denoted by $\Phi_p = \left\{T_{pn}, \ n\in \bbN\right\}$ and $\Phi_s= \left\{T_{sn}, \ n\in \bbN\right\}$, respectively.  The coexisting ad hoc networks under consideration is illustrated in Fig. \ref{fig:blkdiag}, where the red dots represent the primary transmitters and receivers, while the blue dots represent the secondary transmitters and receivers.
We assume that each primary transmitter and receiver has a single antenna, while 
each secondary transmitter has $N$ antennas, and each secondary receiver has $M$ antennas. 
We restrict ourselves to the case when each secondary transmitter transmits only one data stream through its multiple antennas. 

The received signal at the primary receiver $R_{p0}$ is
\begin{eqnarray}\label{rxprim}
y_0 &=& \sqrt{P_p}d_p^{-\alpha/2}h_{00} x_{p0} + \sum_{n: T_{pn} \in \Phi_p \backslash \{T_{p0}\}}\sqrt{P_p}d_{pp, n}^{-\alpha/2}h_{0n}x_{pn} + 
\sum_{n:T_{sn} \in \Phi_s }\sqrt{\frac{P_s}{N}} d_{sp, n}^{-\alpha/2}\bg_{0n}\bu_nx_{sn},
\end{eqnarray}
where $P_p$ and $P_s$ is the transmit power of each primary and secondary transmitter, respectively,
$h_{0n} \in \bbC$ is the channel between $T_{pn}$ and $R_{p0}$, $\bg_{0n}\in \bbC^{1\times  N}$
is the channel between $T_{sn}$ and $R_{p0}$, 
$d_{pp, n}$ and $d_{sp, n}$ is the distance between $T_{pn}$ and $R_{p0}$, and $T_{sn}$ and $R_{p0}$, respectively, 
$\alpha$ is the path loss exponent $\alpha > 2$, $x_{pn}$ and $x_{sn}$  are  data signals transmitted from $T_{pn}$ and $T_{sn}$, respectively, with $x_{pn}, x_{sn} \sim  \ {\cal C}N(0,1)$, $\bu_n \in \bbC ^{N\times 1}$ is the beamformer used by the $n^{th}$ secondary transmitter.  We consider the interference limited regime, i.e.
noise power is negligible compared to the interference power, and drop the additive white Gaussian noise contribution \cite{Weber2005}.
We assume that each $h_{0n}$, and each entry of $\bg_{0n}$ is
i.i.d. ${\cal CN}(0,1)$ to model a richly
scattered fading channel with independent fading coefficients between different
transmitting receiving antennas.

The $ \bbC^{M\times 1}$ received signal $\bv_0$ at the secondary receiver $R_{s0}$  is
\begin{eqnarray}\label{rxcog}
\bv_0 &= & \sqrt{\frac{P_s}{N}} d_{c}^{-\alpha/2}\bQ_{00}\bu_0x_{c0} + 
\sum_{n:T_{sn} \in \Phi_s \backslash \{T_{s0}\}}\sqrt{\frac{P_s}{N}} d_{ss, n}^{-\alpha/2}\bQ_{0n}\bu_nx_{sn} +   \sum_{n:T_{pn} \in \Phi_p}\sqrt{P_p}d_{ps, n}^{-\alpha/2}\bff_{0n}x_{pn},
\end{eqnarray}
where $d_{ss, n}$ and $d_{ps, n}$ is the distance between $T_{sn}$ and $R_{s0}$, and $T_{pn}$ and $R_{s0}$, respectively, 
$\bQ_{0n} \in \bbC^{M\times N}$ is the channel between $T_{sn}$ and $R_{s0}$, $\bff_{0n}\in \bbC^{M\times 1}$
is the channel between $T_{pn}$ and $R_{s0}$.

For partial interference cancelation, the $n^{th}$ secondary receiver multiplies $\bt_n^*$ to the received signal. Thus with signal model (\ref{rxprim}), and (\ref{rxcog}), the signal-to-interference ratio (SIR) for $R_{p0}$  is 
\[SIR_p \bydef\frac{P_pd_p^{-\alpha}|h_{00}|^2}{\sum_{n:T_{pn} \in \Phi_p \backslash \{T_{p0}\}}P_pd_{pp, n}^{-\alpha}|h_{0n}|^2 + 
\sum_{n:T_{sn} \in \Phi_s }P_s d_{sp, n}^{-\alpha}|\bg_{0n}\bu_n|^2},\] and the SIR for $R_{s0}$ is 
\[SIR_s \bydef \frac{P_s d_s^{-\alpha}|\bt_0^*\bQ_{00}\bu_0|^2}{\sum_{n:T_{sn} \in \Phi_s \backslash \{T_{s0}\}}P_sd_{ss, n}^{-\alpha}|\bt_n^*\bQ_{0n}\bu_n|^2 + 
\sum_{n:T_{pn} \in \Phi_p }P_p d_{ps, n}^{-\alpha}|\bt_0^*\bff_{0n}|^2},\] respectively. We assume that $\bt_n = 1$ if $M=1$. 
Without the presence of secondary network, the SIR at the primary receiver $R_{p0}$ is 
$SIR_p^{nc} \bydef\frac{P_pd_p^{-\alpha}|h_{00}|^2}{\sum_{n:T_{pn} \in \Phi_p \backslash \{T_{p0}\}}P_pd_{pp, n}^{-\alpha}|h_{0n}|^2}$.

We assume that the rate of transmission for each primary (secondary) transmitter  is $R_p = \log(1+\beta_p)$ ($R_s = \log(1+\beta_s)$) bits/sec/Hz. Therefore,  a packet transmitted by $T_{p0}$ ($T_{s0}$) can be successfully decoded at $R_{p0}$ ($R_{s0}$), if $SIR_{p} \ge \beta_p$ ($SIR_{s} \ge \beta_s$). 
Without the presence of secondary network, for a given rate $R_p$ bits/sec/Hz, let $\lambda_p$ be the maximum intensity for which the outage probability of the 
primary network $P^{nc}_{p, out}=P\left(SIR_p^{nc} \le \beta_p\right) \le \epsilon_p^{nc}$.  Allowing secondary transmissions increases the interference received at $R_{p0}$ as quantified in $SIR_p$ compared to $SIR_p^{nc}$, and thereby  increases the outage probability from $P_{p, out}^{nc}
$ to $P_{p, out}=P\left(SIR_p \le \beta_p\right) $ for a fixed $\lambda_p$. Let the increased outage probability tolerance at the primary receivers be  $\epsilon_p^{nc} +\Delta_p$. Then we want to find the 
maximum intensity of secondary transmitters $\lambda_s$ for which  $P_{p, out} \le \epsilon_p^{nc} + \Delta_p$, and the outage probability of the secondary network  $P_{s, out} =P\left(SIR_s \le \beta_s\right) \le \epsilon_s$. Thus, the maximum intensity of the secondary network is $ \lambda_s^{\star} = \max_{P_{p, out}\le \epsilon_p^{nc}+\Delta_p, \ P_{s, out}\le \epsilon_s}
\lambda$.
 Hence following \cite{Weber2005} the transmission capacity of the secondary network is 
$C_s \bydef \lambda_s^{\star} (1-\epsilon_s) R_s$ bits/sec/Hz/m$^2$. In the rest of the paper, we derive $\lambda_s^{\star}$ with or without multiple antennas at the secondary nodes.
Following \cite{Weber2005}, to compute the outage probability $P_{p,out}$ and $P_{s,out}$, we consider a typical
transmitter receiver pair $(T_{p0}, R_{p0})$ and $(T_{s0}, R_{s0})$, respectively. 
%

\section{$N=1, M=1$}
\label{sec:siso}

\begin{thm}\label{thm:sisotol} With $M=N=1$, and $c_1 = \frac{2\pi^2 Csc(\frac{2\pi}{\alpha})}{\alpha}$ where $Csc$ is co-secant,  \[\lambda^{\star}_s = \min\left\{\frac{-\ln\left( \frac{1-\epsilon_p^{nc} -\Delta_p}{1-\epsilon_p^{nc}}\right)}{d_p^2} \left(\frac{P_p}{P_s\beta_p}\right)^{\frac{\alpha}{2}} , \frac{-\ln(1-\epsilon_s) - \lambda_p c_1d_s^{2} 
(\frac{P_s\beta_s}{P_s})^{\frac{2}{\alpha}}}{c_1\beta_s^{\frac{2}{\alpha}}d_s^{2}} \right\}.\] 
\end{thm}
\begin{proof} See appendix \ref{app:siso}.
\end{proof}
{\it Discussion:} In this section we derived an exact expression for the intensity of the secondary network, when the secondary network operates under an outage constraint at the primary receivers. In prior work, an exact expression for the intensity of the secondary network was derived in \cite{CuiCog2009} only for $\alpha=4$, while upper and lower bounds  were derived in \cite{HuangCog2009}. Our derivation of the exact expression used the fact that the interference caused by the primary and the secondary transmitters at either the primary or the secondary receiver is
independent. Using the independence, we then applied the Laplace transform method of \cite{Baccelli2006}, to derive the exact expression.
Using the derived expression, we characterized the increase in the intensity of the secondary network by  increasing the allowed outage probability tolerance at the primary receivers.

\section{Multiple Transmit Antennas $N$, Single Receive Antenna $M=1$}
\label{sec:miso}
In this section we consider the case when each secondary transmitter has $N$ antennas, while each secondary receiver has a single antenna, $M=1$. With multiple antennas at each secondary transmitter, 
two of the promising strategies to increase the intensity of the secondary network are:
1) nulling the interference caused to the primary receivers, or 
2) transmit beamforming on the channel to its corresponding receiver.
The first strategy decreases the interference received by each primary receiver, and hence increases the first term inside the minimum for the expression of $\lambda_s^{\star}$ (Theorem \ref{thm:sisotol}), while the second strategy helps in increasing the signal power for each secondary transmission and increases the second term inside the minimum for the expression of $\lambda_s^{\star}$ (Theorem \ref{thm:sisotol}). Since the transmission capacity is the minimum of the transmission capacity 
while considering the outage constraint at the primary and secondary receivers (Theorem \ref{thm:sisotol}), multiple antennas at each secondary transmitter should be used to jointly increase the  transmission capacity under both the constraints. 
Towards that end, 
we assume that out of the total $N$ STDOF,  $k$ are used for nulling interference towards the $k$ nearest primary receivers, while the rest $N-k$ are used for transmit beamforming.
 Note that nulling interference towards the $k$ nearest primary receivers, does not ensure that the interference contribution from the $k$ nearest secondary interferers are canceled at each primary receiver. This is illustrated in Fig. \ref{fig:nearcanc}, where the red dots represent the primary transmitters and receivers, while the blue dots represent the secondary transmitters and receivers, and each secondary transmitter nulls its interference towards its $k$ primary receivers. The interference nulling, however, does not necessarily nulls the interference from the $k$ nearest secondary transmitters  at each primary receiver, e.g.  
 at primary receiver $R_{pl}$ (circle $\bB$), interference is nulled from only one nearest 
 secondary transmitter.
 
 Let $C$ be the random variable denoting the  number of nearest secondary interferers that are canceled at any primary receiver. Then with $C=c$ nearest secondary interferers canceled at $R_{p0}$, the SIR at $R_{p0}$ is 
\[SIR_p \bydef\frac{P_pd_p^{-\alpha}|h_{00}|^2}{\sum_{n:\ T_{pn} \in \Phi_p \backslash \{T_{p0}\}}P_pd_{pp, n}^{-\alpha}|h_{0n}|^2 + 
\sum_{n: \ n>c, \ T_{sn} \in \Phi_s }P_s d_{sp, n}^{-\alpha}|\bg_{0n}\bu_n|^2}.\]  The SIR at $R_{s0}$ is 
\[SIR_s \bydef \frac{P_s d_s^{-\alpha}|\bq_{00}\bu_0|^2}{\sum_{n:\ T_{sn} \in \Phi_s \backslash \{T_{s0}\}}P_sd_{ss, n}^{-\alpha}|\bq_{0n}\bu_n|^2 + 
\sum_{n:\ T_{pn} \in \Phi_p }P_p d_{ps, n}^{-\alpha}|f_{0n}|^2},\] where $\bq_{0n}$ is the $1\times N$ channel vector between 
$T_{sn}$ and $R_{s0}$, and $\bu_n$ lies in the null space of  $[\bg_{1n}^T \ldots \bg_{kn}^T]$  to null the interference towards the $k$ nearest  primary receivers, and chosen such that it maximizes the signal power $|\bq_{nn}\bu_n|^2$. From \cite{Jindal2008a}, $\bu_n = \frac{\bq_{nn}^{*}\bS\bS^*}{|\bq_{nn}^{*}\bS\bS^*|}$, where $\bS \in \bbC^{N\times N-k}$ is the orthonormal basis of the null space of $[\bg_{1n}^T \ldots \bg_{kn}^T]$.

\begin{lemma}\label{lem:sisosigdist}
The  signal power $s\bydef |\bq_{00}\bu_0|^2$ at the secondary receiver with $\bu_0 = \frac{\bq_{00}^{*}\bS\bS^*}{|\bq_{00}^{*}\bS\bS^*|}$ is distributed as Chi-square with $2(N-k)$ DOF The interference power at the secondary receiver from the secondary transmitter $n$, 
$|\bq_{0n}\bu_n|^2$, is distributed as Chi-square with $2$ degrees of freedom (DOF).
\end{lemma}
\begin{proof} The first statement follows from \cite{Jindal2008a}. The second statement follows since $\bu_n$ and $\bq_{0n}$ are independent and since 
each entry of $\bq_{0n} \sim {\cal CN}(0,1)$.
\end{proof}

\begin{lemma}\label{lem:sisointdist} 
The interference  power at primary receiver from the secondary transmitter $n$, 
$|\bg_{0n}\bu_n|^2$, is distributed as Chi-square with $2$ DOF.
\end{lemma}
\begin{proof} Follows from the fact that $\bu_n$ and $\bg_{0n}$ are independent and since each entry of $\bg_{0n} \sim {\cal CN}(0,1)$.
\end{proof}

\begin{lemma}\label{lem:sumintf}
The interference received at any secondary receiver from the union of transmitters belonging to $\Phi_p$ (with intensity $\lambda_p$, transmission power $P_p$) and $\Phi_s$ (with intensity $\lambda_s$, transmission power $P_s$)  is equal to the interference received from transmitters belonging to a single PPP $\Phi$ with intensity 
$\lambda_p P_p^{\frac{2}{\alpha}} + \lambda_s P_s^{\frac{2}{\alpha}}$ and unit transmission power.
\end{lemma}

\begin{proof} See appendix \ref{app:misomisc}.
\end{proof}



%
\begin{thm}
\label{thm:Ntx1}
With $M=1$, when each secondary transmitter uses $k$ DOF for nulling its  interferers towards its $k$ nearest primary receivers, and $N-k$ DOF for beamforming, $k=\theta N, \theta \in (0,1]$ optimizes the scaling of  intensity of the secondary network $\lambda_s$, and the intensity scales as 
$\lambda_s^{\star} =  \Omega\left(\min\{N^{\frac{2}{\alpha}}, N^{1-\frac{2}{\alpha}} \}\right)$, and $\lambda_s^{\star} = {\cal O}\left(N^{\frac{2}{\alpha}}\right)$.
\end{thm}
\begin{proof} See appendix \ref{app:miso}.
\end{proof}
{\it Discussion:} In this section we showed that using $k=\theta N, \ \theta \in (0,1],$ STDOF maximizes the intensity of the secondary network, and the  lower bound on the intensity of the secondary network scales sublinearly in $N$. A major obstacle in the analysis stems from the fact that when each secondary transmitter nulls its interference towards its $k$ nearest primary receivers, it does not imply that the interference from the $k$ nearest secondary interferers is canceled at any primary receiver. Therefore, the
results of this section do not follow directly from previous work on finding the intensity of ad-hoc networks with multiple antennas when each receiver cancels interference from some of its nearest interferers \cite{Vaze2009, Jindal2008a}.

Without the presence of the secondary network, the intensity scales  as $N^{\frac{2}{\alpha}}$ when each transmitter has $N$ antennas and uses transmit beamforming \cite{Hunter2008}.
Comparing our results with \cite{Hunter2008}, we find that in our case the
intensity of the secondary network scales as $\min\left\{N^{\frac{2}{\alpha}}, N^{1-\frac{2}{\alpha}}\right\}$ because the secondary transmitters have to satisfy two outage constraints: one at the primary receiver and the other at the secondary receiver.  Since each primary transmitter and receiver has only a single antenna, even in the best case when exactly $k$ nearest secondary interferers are canceled at each primary receiver, considering the outage constraint at the primary receivers, the intensity of the secondary network scales at best as $N^{1-\frac{2}{\alpha}}$.

\section{Multiple transmit and receive antennas}
\label{sec:mimo}
In this section we assume that each secondary transmitter has $N$ antennas and each secondary receiver has $M$ antennas. The $N$ transmit antennas at each secondary transmitter are used to null interference towards its $N-1$ nearest primary receivers,\footnote {For analytical tractability we do not consider the general case of using $k$ STDOF for nulling and rest $N-k$ for beamforming.} while 
each secondary receiver uses its $m$ SRDOF for canceling the nearest interferers from the union of the primary and the secondary interferers, and the rest $N-m$ SRDOF are used for increasing the strength of signal of interest. 

Then,  the SIR at $R_{p0}$ is 
\[SIR_p \bydef\frac{P_pd_p^{-\alpha}|h_{00}|^2}{\sum_{n:T_{pn} \in \Phi_p \backslash \{T_{p0}\}}P_pd_{pp, n}^{-\alpha}|h_{0n}|^2 + 
\sum_{n:n>c, \ T_{sn} \in \Phi_s }P_s d_{sp, n}^{-\alpha}|\bg_{0n}\bu_n|^2},\]  
and the SIR at $R_{s0}$ is 
\[SIR_s \bydef \frac{P_s d_s^{-\alpha}|\bt_0^{*}\bQ_{00}\bu_0|^2}{\sum_{n:T_{sn} \in \Phi_s \backslash \{T_{s0}\}}P_sd_{ss, n}^{-\alpha}|\bt_n^{*}\bQ_{0n}\bu_n|^2 + 
\sum_{n:T_{pn} \in \Phi_p }P_p d_{ps, n}^{-\alpha}|\bt_0^*\bff_{0n}|^2},\] where  $\bu_n$ lies in the null space of  $[\bg_{1n}^T \ldots \bg_{N-1n}^T]$  to null the interference towards its $N-1$ nearest  primary receivers,  $\bt_n$ lies in the null space of channel vectors corresponding to its 
$m$ nearest interferers from $\{\Phi_p \cup \Phi_s\}\backslash \{T_{sn}\}$  chosen such that it maximizes the signal power 
$|\bt^*\bQ_{nn}\bu_n|^2$. 
From \cite{Jindal2008a}, $\bt_n = \frac{(\bQ_{nn}\bu)^{*}\bR\bR^*}{|(\bQ_{nn}\bu)^{*}\bR\bR^*|}$, where $\bR \in \bbC^{M\times M-m}$ is the orthonormal basis of the null space of channel vectors corresponding to its 
$m$ nearest interferers from $\Phi_p \cup \Phi_s \backslash \{T_{sn}\}$.

\begin{lemma}
\label{lem:sigpowmimo}
The  signal power $s\bydef |\bt_0^*\bQ_{00}\bu_0|^2$ at the secondary receiver with  $\bt_n = \frac{(\bQ_{nn}\bu)^{*}\bR\bR^*}{|(\bQ_{nn}\bu)^{*}\bR\bR^*|}$ is distributed as Chi-square with $2(M-m)$ DOF The interference power at secondary receiver from the secondary transmitter $n$ 
$I_{ss}^{0n} \bydef |\bt_0^*\bq_{0n}\bu_n|^2$, and the interference power at secondary receiver from the primary transmitter $n$ 
$I_{ps}^{0n} \bydef |\bt_0^*\bff_{0n}|^2$ is distributed as Chi-square with $2$ DOF.
\end{lemma}
\begin{proof} The first statement follows from \cite{Jindal2008a}. The second and third statement follows since $\bt_0^*$, $\bu_n$, and  $\bq_{0n}$ are independent, and since each entry of $\bq_{0n}, \bff_{0n} \sim {\cal CN}(0,1)$.
\end{proof}

\begin{lemma}
\label{lem:2repcogintf}
The interference received at the typical secondary receiver $R_{s0}$
\[\sum_{n:T_{sn} \in \Phi_s \backslash \{T_{s0}\}}P_sd_{ss, n}^{-\alpha}I^{0n}_{ss} + 
\sum_{n: T_{pn} \in \Phi_p }P_p d_{ps, n}^{-\alpha}I_{ps}^{0n} =  
\sum_{n: T \in \Phi \backslash \{T_{s0}\}}P_nd_{n}^{-\alpha}I^{0n}, \]
where $I^{0n}$ is Chi-square distributed with $2$ DOF, $\Phi = \{\Phi_s \cup \Phi_p\}$, and $P_n$ is a binary random variable which takes value $P_p$ with probability $\frac{\lambda_p}{\lambda_p+\lambda_s}$, and value $P_s$ with probability  $\frac{\lambda_s}{\lambda_p+\lambda_s}$.
\end{lemma}

\begin{proof} 
Since the superposition of two PPP's is a PPP, consider the union of $\Phi_p$ and  $\Phi_s$ as a single PPP $\Phi = \{\Phi_s \cup \Phi_p\}$. Thus, the interference received at the typical secondary receiver $R_{s0}$ is derived from the transmitters corresponding to $\Phi$ with channel gains $I_{ss}^{0n}$ or $I_{ps}^{0n}$, where both $I_{ss}^{0n}$ and $I_{ps}^{0n}$ are distributed as Chi-square with $2$ DOF. Note that the primary transmitters use power $P_p$, and the secondary transmitters use power $P_s$. The probability that any randomly chosen node of $\Phi$ belongs to $\Phi_p$ is 
$\frac{\lambda_p}{\lambda_p+\lambda_s}$ (Lemma $3$) \cite{HuangCog2009}, hence the power transmitted by any node of $\Phi$ is $P_p$ with probability $\frac{\lambda_p}{\lambda_p+\lambda_s}$, and $P_s$ with probability $\frac{\lambda_s}{\lambda_p+\lambda_s}$.
\end{proof}

\begin{thm}
When each secondary transmitter uses $N-1$ DOF for nulling, and each secondary receiver uses $m$  DOF for  canceling the $m$ nearest interferers from $\{\Phi_s \cup \Phi_p\} \backslash \{T_{s0}\}$, then $m=\theta N, \ \theta \in (0,1])$ maximizes the lower and upper bound on the  intensity of secondary ad-hoc network, and  
 $\lambda_s^{\star} =  \Omega\left(\min\{M, N^{1-\frac{2}{\alpha}} \}\right) $, and $\lambda_s^{\star} = {\cal O}\left(\min\left\{N, M^{1+\frac{2}{\alpha}}\right\}\right)$.
\end{thm}
\begin{proof} See appendix \ref{app:mimo}.
\end{proof}

 {\it Discussion:} In this section we showed that using a fraction of total SRDOF maximizes the scaling of the intensity  of the secondary network, when $N-1$ STDOF are used for interference nulling by each secondary transmitter. Comparing results of this Section with Section \ref{sec:miso}, we observe that employing similar number of  antennas at both the secondary transmitters and receivers in comparison to having multiple antennas only at the secondary transmitters improves the intensity scaling for path-loss exponent $\alpha >4$.

Without the presence of the secondary network, with $M$ receive antennas, the intensity of an ad hoc network  is shown to scale linearly with $M$ even with a single transmit antenna \cite{Vaze2009, Jindal2008a}. 
In contrast, with coexisting networks, our results show that with a single transmit antenna, the intensity  of the secondary network is independent of the number of receive antennas of each secondary receiver. This result can be understood by noting that with  a single transmit antenna, none of the secondary transmitters can null their interference towards any of the primary receivers, and the intensity  of the secondary
network is limited by the outage constraint at the primary receivers, and hence independent of the number of receive antennas of the secondary receivers. Therefore it is imperative to use similar number of transmit and receive antennas at the secondary nodes to maximize the intensity  of the secondary network.

\section{Simulations}
\label{sec:sim}
In all the simulation results we use  $\alpha =3$,  $d_p=d_s=1m$, $\frac{P_p}{P_s}=2$, $\beta _p=\beta_s=1$ corresponding to 
$R_p= R_s=1$ bits/sec/Hz.  In Fig. \ref{fig:svp}, we plot the transmission capacity of the secondary network with respect to the  transmission capacity of the primary  network to show how one can be tradeoff against other for $\epsilon_p^{nc}+\Delta_p =\epsilon_s=.1$ and $M=N=1$. In Fig. \ref{fig:svep} we plot the transmission capacity of the secondary network with increasing primary outage probability constraint for a primary network intensity of $\lambda_p = 0.01$, secondary outage constraint $\epsilon_s=.1$, and $M=N=1$. We see that the transmission capacity of the secondary network increases with  increasing primary outage probability constraint until a point where the secondary outage constraint becomes tight, and thereafter the transmission capacity of the secondary network is constant, limited by the secondary outage constraint.
In Fig. \ref{fig:smiso}, we plot the transmission capacity of the secondary network with respect to the number of transmit antennas $N$ at each secondary transmitter for different values of $\theta= \frac{1}{2}, \frac{1}{3}, \frac{1}{4}$, with $M=1$. We see that for all the values of $\theta$ the  transmission capacity of the secondary network  scales sublinearly with $N$ as predicted by Theorem $2$. In Fig. \ref{fig:smimo}, we plot the transmission capacity of the secondary network with respect to the number of secondary transmit and receive antennas $N$ and $M$. We see that for $N=M$ the  transmission capacity of the secondary network  scales sublinearly with $N$, however, for $N=1$ the transmission capacity of the secondary network is constant and determined by the outage constraint of the primary network.

\section{Conclusions}\label{sec:conc}
In this paper we considered deployment of a secondary network overlaid on top of an existing primary network. 
Under an outage constraint at the primary network's receivers from the secondary network's transmitters, we characterized the maximum 
 intensity of the secondary network with and without multiple antennas. We showed that with a single antenna at the secondary nodes, the  intensity  of the secondary network grows logarithmically with the outage probability tolerance at the primary receivers. With multiple antennas, we characterized the optimal role of multiple antennas at the secondary nodes that maximizes the scaling of the 
intensity of the secondary network. We showed that employing multiple antennas only at the secondary receivers does not yield any gain. To exploit the multiple antenna gain, either the multiple antennas should be employed at the secondary transmitters, or at both the secondary transmitters and receivers. We showed that with multiple antennas only at the secondary transmitters, the intensity  of the secondary network scales sublinearly with the number of antennas. The sublinear scaling of the intensity cannot be improved by employing multiple antennas at both the secondary transmitters and receivers, however, the sublinear exponent is better for path-loss exponent greater than four.

\appendices

\section{}\label{app:siso}
First we find the value of $\lambda_s$ for which $P_{p,out} \le \epsilon_p^{nc}+\Delta_p$.
Using the definition of $SIR_p$, \begin{eqnarray*}
P_{p, out} &=& P\left(\frac{P_pd_p^{-\alpha}|h_{00}|^2}{\sum_{n: \ T_{pn} \in \Phi_p \backslash \{T_{p0}\}}P_pd_{pp, n}^{-\alpha}|h_{0n}|^2 + 
\sum_{n: \ T_{sn} \in \Phi_s }P_s d_{sp, n}^{-\alpha}|g_{0n}|^2} \le \beta_p\right),\\
\epsilon_p^{nc} + \Delta_p &=& P\left(\frac{P_pd_p^{-\alpha}|h_{00}|^2}{\sum_{n: \ T_{pn} \in \Phi_p \backslash \{T_{p0}\}}P_pd_{pp, n}^{-\alpha}|h_{0n}|^2 + 
\sum_{n: \ T_{sn} \in \Phi_s }P_s d_{sp, n}^{-\alpha}|g_{0n}|^2} \le \beta_p\right), \\
&\stackrel{(a)}= & 1- \bbE_{I_{pp}, I_{sp}} \left\{\exp\left(-\frac{\beta_p(P_p I_{pp} + P_sI_{sp})d_p^{\alpha}}{P_p}\right)\right\},\\
&\stackrel{(b)}= & 1- \int_{0}^{\infty} \exp\left(-\beta_p(I_{pp}=s )d_p^{\alpha}\right) f_{I_{pp}}(s) ds
\int_{0}^{\infty} \exp\left(-\frac{\beta_p( I_{sp}=t)P_sd_p^{\alpha}}{P_p}\right) f_{I_{sp}}(t) dt,
\end{eqnarray*}
\begin{eqnarray*}
&\stackrel{(c)}= & 1-{\cal L}_{I_{pp}}\left(\beta_pd_p^{\alpha}\right)  {\cal L}_{I_{sp}}\left(\frac{P_s\beta_pd_p^{\alpha}}{P_p}\right), \\
&\stackrel{(d)}= & 1-\exp\left(-\lambda_p c_1\beta_p^{\frac{2}{\alpha}}d_p^{2}\right)\exp\left(-\lambda_s c_1d_p^2\left(\frac{P_s\beta_p}{P_p}\right)^{\frac{2}{\alpha}}\right),
\end{eqnarray*}
where $(a)$ follows by letting $ I_{pp} \bydef \sum_{n:\ T_{pn} \in \Phi_p \backslash \{T_{p0}\}}d_{pp, n}^{-\alpha}|h_{0n}|^2$, and $I_{sp} \bydef \sum_{n:\ T_{sn} \in \Phi_s } d_{sp, n}^{-\alpha}|g_{0n}|^2$, and taking the expectation with respect to $|h_{00}|^2$ since $|h_{00}|^2$ is exponentially distributed, $(b)$ follows since $I_{pp}$ and $I_{sp}$ are independent, $(c)$ follows by defining 
${\cal L}_{I}(.)$ as the Laplace transform of $I$, and $(d)$ follows since ${\cal L}_{I_{pp}}\left(\beta_pd_p^{\alpha}\right) = \exp\left(-\lambda_p c_1\beta_p^{\frac{2}{\alpha}}d_p^{2}\right)$ \cite{Baccelli2006}, where $c_1$ is a constant.
Note that $\epsilon_p^{nc} = 1-\exp\left(-\lambda_p c_1\beta_p^{\frac{2}{\alpha}}d_p^{2}\right)$ \cite{Baccelli2006}, hence 
$\lambda_s = \frac{-\ln\left( \frac{1-\epsilon_p^{nc} -\Delta_p}{1-\epsilon_p^{nc}}\right)}{d_p^2} \left(\frac{P_p}{P_s\beta_p}\right)^{\frac{\alpha}{2}}$.

Next, we evaluate the value of $\lambda_s$ such that $P_{s,out} \le \epsilon_s$. By definition
\begin{eqnarray*}
P_{s, out} &=& P\left(SIR_s \le \beta_s\right), \\
& = & P\left(\frac{P_s d_s^{-\alpha}|q_{00}|^2}{\sum_{n:\ T_{sn} \in \Phi_s \backslash \{T_{s0}\}}P_sd_{ss, n}^{-\alpha}|q_{0n}|^2 + 
\sum_{n:\ T_{pn} \in \Phi_p }P_p d_{ps, n}^{-\alpha}|f_{0n}|^2} \le \beta_s\right). 
\end{eqnarray*}
Using the same analysis as above for finding $\lambda_s$ such that $P_{p,out} \le \epsilon_p^{nc}+\Delta_p$ we get 
$\lambda_s  = \frac{-\ln(1-\epsilon_s) - \lambda_p c_1d_s^{2} 
(\frac{P_s\beta_s}{P_s})^{\frac{2}{\alpha}}}{c_1\beta_s^{\frac{2}{\alpha}}d_s^{2}}$.

\section{}\label{app:misomisc}The interference received at secondary receiver $R_{s0}$ is 
\begin{eqnarray*}
I_{ss}+ I_{ps} & = &  \sum_{n:\ T_{sn} \in \Phi_s \backslash \{T_{s0}\}}P_sd_{ss, n}^{-\alpha}|\bq_{0n}\bu_n|^2 + 
\sum_{n:\ T_{pn} \in \Phi_p }P_p d_{ps, n}^{-\alpha}|f_{0n}|^2,\\
&=& \sum_{n:\ T_{sn} \in \Phi_s \backslash \{T_{s0}\}}\left(\frac{d_{ss, n}}{P_s^{\frac{1}{\alpha}}}\right)^{-\alpha}|\bq_{0n}\bu_n|^2 + 
\sum_{n:\ T_{pn} \in \Phi_p }\left(\frac{d_{ps, n}}{P_p^{\frac{1}{\alpha}}}\right)^{-\alpha} |f_{0n}|^2, \\
&\stackrel{(a)}=&  \sum_{n:\ T_{sn} \in \Phi'_c \backslash \{T_{s0}\}} \left(d^{'}_{ss,n}\right)^{-\alpha}|\bq_{0n}\bu_n|^2 + 
\sum_{n:\ T_{pn} \in \Phi^{'}_p }\left(d^{'}_{ps,n}\right)^{-\alpha} |f_{0n}|^2,\\
&\stackrel{(b)}=& \sum_{n:\ T_{sn} \in \Phi_{sp} \backslash \{T_{s0}\}} \left(d^{'}_{cc,n}\right)^{-\alpha}|\bq_{0n}\bu_n|^2,
\end{eqnarray*}
where in (a) $\Phi'_c$ is a PPP with intensity $\lambda_s P_s^{\frac{2}{\alpha}}$, and $\Phi'_p$ is a PPP with intensity $\lambda_p P_p^{\frac{2}{\alpha}}$, since scaling the distances in a PPP by $\frac{1}{a}$ increases the intensity of a PPP by $a^2$, and $(b)$ follows by defining $\Phi_{sp} = \{\Phi_p \cup \Phi_s\}$ with intensity $\lambda_p P_p^{\frac{2}{\alpha}} + \lambda_s P_s^{\frac{2}{\alpha}}$ since superposition of two PPP's is a PPP with intensity equal to the sum of two superposed PPP's,  and since $\bq_{0n}\bu_n$ and $f_{0n}$ are identically distributed from Lemma \ref{lem:sisosigdist}.

\section{}\label{app:miso}
Since we are interested in establishing the scaling behavior of the intensity of the secondary network with respect to $N$, we consider the case when both $N$ and $k$ are large enough.
First we find the value of $\lambda_p$ for which $P_{p,out } \le \epsilon_p^{nc}+\Delta_p$.

{\bf Lower Bound:}
The outage probability $P_{p, out}$ is 
\begin{eqnarray*}
P_{p, out} &=& \bbE_{C}\left\{P\left(\frac{P_pd_p^{-\alpha}|h_{00}|^2}{\sum_{n:\ T_{pn} \in \Phi_p \backslash \{T_{p0}\}}P_pd_{pp, n}^{-\alpha}|h_{0n}|^2 + 
\sum_{n:\ n>c, T_{sn} \in \Phi_s }P_s d_{sp, n}^{-\alpha}|g_{0n}|^2} \le \beta_p\right)\right\},\\
\epsilon_p + \Delta_p &=& \bbE_{C}\left\{P\left(\frac{P_pd_p^{-\alpha}|h_{00}|^2}{\sum_{n:\ T_{pn} \in \Phi_p \backslash \{T_{p0}\}}P_pd_{pp, n}^{-\alpha}|h_{0n}|^2 + 
\sum_{n:\ n>c, T_{sn} \in \Phi_s }P_s d_{sp, n}^{-\alpha}|g_{0n}|^2} \le \beta_p\right) | C< \left\lfloor k/m\right\rfloor\right\}\\
&&\times P\left(C< \left\lfloor k/m\right\rfloor\right) \\
&&+ \bbE_{C}\left\{P\left(\frac{P_pd_p^{-\alpha}|h_{00}|^2}{\sum_{n:\ T_{pn} \in \Phi_p \backslash \{T_{p0}\}}P_pd_{pp, n}^{-\alpha}|h_{0n}|^2 + 
\sum_{n:\ n>c, T_{sn} \in \Phi_s }P_s d_{sp, n}^{-\alpha}|g_{0n}|^2} \le \beta_p\right) | C\ge \left\lfloor k/m\right\rfloor\right\} \\
&&\times P\left(C\ge  \left\lfloor k/m\right\rfloor\right), \\
&\stackrel{(a)}\le & \delta + \bbE_{C}\left\{1- \bbE_{I_{pp}, I_{sp}^c} \left\{\exp\left(-\frac{\beta_p(P_p I_{pp} + P_sI^c_{sp})d_p^{\alpha}}{P_p}\right)\right\}| C\ge \left\lfloor k/m\right\rfloor\right\},\\
&\stackrel{(b)}= & \delta + \bbE_{C}\left\{ 1- \int_{0}^{\infty} \exp\left(-\beta_p(I_{pp}=s )d_p^{\alpha}\right) f_{I_{pp}}(s) ds \right. \\
&& \left.
 \int_{0}^{\infty} \exp\left(-\frac{\beta_p( I^x_{sp}=t)P_sd^{\alpha}}{P_p}\right) f_{I^x_{sp}}(t) dt |  C\ge \left\lfloor k/m\right\rfloor\right\},\\
&\stackrel{(c)}= & \delta + \bbE_{C}\left\{ 1-{\cal L}_{I_{pp}}\left(\beta_pd_p^{\alpha}\right) \left(1-P\left(\frac{P_pd_p^{-\alpha}|h_{00}|^2}{ 
\sum_{n:\ n>c, T_{sn} \in \Phi_s }P_s d_{sp, n}^{-\alpha}|g_{0n}|^2} \le \beta_p\right)\right) | C\ge \left\lfloor k/m\right\rfloor\right\}, \\
&\stackrel{(d)}\le &  \delta + \bbE_{C}
\left\{  1-\exp\left(-\lambda_p c_1\beta_p^{\frac{2}{\alpha}}d_p^{2}\right)
\left(1-(\pi \lambda_s)^{\frac{\alpha}{2}}\beta_p \left(\frac{P_s}{P_p}\right)d_p^{\alpha}\phi\left(\left(\frac{\alpha}{2}-1\right)^{-1}(c+1)^{1-\frac{\alpha}{2}} + c_3\right)\right) \right.\\
&& \left.
 | C\ge \left\lfloor k/m\right\rfloor\right\},\\
&= & \delta + 1-\exp\left(-\lambda_p c_1\beta_p^{\frac{2}{\alpha}}d_p^{2}\right)\\
&&+\exp\left(-\lambda_p c_1\beta_p^{\frac{2}{\alpha}}d_p^{2}\right)
(\pi \lambda_s)^{\frac{\alpha}{2}}\beta_p \left(\frac{P_s}{P_p}\right)d_p^{\alpha}\phi\left(\left(\frac{\alpha}{2}-1\right)^{-1}\bbE_{C}\left\{(c+1)^{1-\frac{\alpha}{2}}| C\ge \left\lfloor k/m\right\rfloor\right\} + c_3\right),\\
&\stackrel{(e)}\le &\delta+ \epsilon_p +\exp\left(-\lambda_p c_1\beta_p^{\frac{2}{\alpha}}d_p^{2}\right)
(\pi \lambda_s)^{\frac{\alpha}{2}}\beta_p \left(\frac{P_s}{P_p}\right)d_p^{\alpha}\phi\left(\left(\frac{\alpha}{2}-1\right)^{-1}(\left\lfloor k/m\right\rfloor+1)^{1-\frac{\alpha}{2}} + c_3\right),\end{eqnarray*}
where $(a)$ follows by letting $m\in \bbN$  such that $P\left(C< \left\lfloor k/m\right\rfloor\right) \le \delta, \ \delta \le \Delta_p$, and $m$ is independent of $k$, and $I_{sp}^c \bydef \sum_{n:\ n>c, T_{sn} \in \Phi_s } d_{sp, n}^{-\alpha}|g_{0n}|^2$, and $|h_{00}|^2$ is exponentially distributed, existence of $m\in \bbN$ such that $P\left(C< \left\lfloor k/m\right\rfloor\right) \le \Delta_p$ is guaranteed, since for large values of $k$ canceling only a few nearest secondary interferers has a very small probability, $(b)$ follows since $I_{pp}$ and $I_{sp}$ are independent, $(c)$ follows by defining 
${\cal L}_{I}(.)$ as the Laplace transform of $I$, and $(d)$ follows from the upper bound on outage probability \cite[Theorem 4]{Vaze2009}, $(e)$ follows since 
for $C\ge \left\lfloor k/m\right\rfloor, \bbE_{C}\left\{(c+1)^{1-\frac{\alpha}{2}}| C\ge \left\lfloor k/m\right\rfloor\right\} \le  (k+1)^{1-\frac{\alpha}{2}}$ for $\alpha >2$, and $\epsilon_p = 1-\exp\left(-\lambda_p c_1\beta_p^{\frac{2}{\alpha}}d_p^{2}\right)
$. Thus, 
$\lambda_s \ge \frac{1}{\pi}
\left(\frac{\Delta_p-\delta}{\exp\left(-\lambda_p c_1\beta_p^{\frac{2}{\alpha}}\right)d_p^{2}\beta_p \left(\frac{P_s}{P_p}\right)d_p^{\alpha}\phi
\left(\left(
\frac{\alpha}{2}-1\right)^{-1} (\left\lfloor k/m\right\rfloor+1)^{1-\frac{\alpha}{2}}+ c_3\right)^{\frac{2}{\alpha}}
} \right)$, 
and $\lambda_s = \Omega\left( k^{1-\frac{2}{\alpha}}\right)$.

{\bf Upper bound:}
To find an upper bound on $\lambda_s$, we consider the case when exactly $k$ nearest secondary interferers are canceled at each primary receiver using $k$ DOF for nulling by each secondary transmitter. This gives an upper bound since in general the the number of nearest interferers  canceled at each primary receiver is a random variable, and the performance is limited by those primary receivers that have less than $k$ nearest interferers canceled.
Thus, 
\begin{eqnarray*}
P_{p, out} &=& \bbE_{C}\left\{P\left(\frac{P_pd_p^{-\alpha}|h_{00}|^2}{\sum_{n:T_{pn} \in \Phi_p \backslash \{T_{p0}\}}P_pd_{pp, n}^{-\alpha}|h_{0n}|^2 + 
\sum_{n>c, T_{sn} \in \Phi_s }P_s d_{sp, n}^{-\alpha}|g_{0n}|^2} \le \beta_p\right)\right\},\\
\epsilon_p + \Delta_p &\stackrel{(a)}\ge&P\left(\frac{P_pd_p^{-\alpha}|h_{00}|^2}{\sum_{n:T_{pn} \in \Phi_p \backslash \{T_{p0}\}}P_pd_{pp, n}^{-\alpha}|h_{0n}|^2 + 
\sum_{n>k, T_{sn} \in \Phi_s }P_s d_{sp, n}^{-\alpha}|g_{0n}|^2} \le \beta_p\right) \\
&= & 1- \bbE_{I_{pp}, I_{sp}^x} \left\{\exp\left(-\frac{\beta_p(P_p I_{pp} + P_sI^k_{sp})d_p^{\alpha}}{P_p}\right)\right\},\\
&= &  1- \int_{0}^{\infty} \exp\left(-\beta_p(I_{pp}=s )d_p^{\alpha}\right) f_{I_{pp}}(s) ds   \int_{0}^{\infty} \exp\left(-\frac{\beta_p( I^k_{sp}=t)P_sd^{\alpha}}{P_p}\right) f_{I^x_{sp}}(t) dt,\\
&= &1-{\cal L}_{I_{pp}}\left(\beta_pd_p^{\alpha}\right) \left(1-P\left(\frac{P_pd_p^{-\alpha}|h_{00}|^2}{ 
\sum_{n>k, T_{sn} \in \Phi_s }P_s d_{sp, n}^{-\alpha}|g_{0n}|^2} \le \beta_p\right)\right) , \\
&\stackrel{(b)}\ge & 1-\exp\left(-\lambda_p c_1\beta_p^{\frac{2}{\alpha}}d_p^{2}\right)\frac{\left(k+\frac{5}{8}+\frac{\alpha}{4}\right)^{\frac{\alpha}{2}}}{d^{\alpha}\beta (\pi\lambda)^{\frac{2}{\alpha}}}, \end{eqnarray*}
 where $(a)$ follows from the fact that canceling exactly $k$ nearest secondary interferers at each primary receiver using $k$ DOF for nulling by each secondary transmitter provides the best performance, and $(b)$ follows from the lower  bound on outage probability \cite[Theorem 4]{Vaze2009}.
 Thus, from the lower and upper bound
 \begin{equation}\label{eq:lbubmiso}
 \lambda_s = {\cal O}(k), \text{ and} \ \lambda_s = \Omega(k^{1-\frac{2}{\alpha}}).
\end{equation}
Next, we evaluate the maximum $\lambda_s$ such that  $P_{s,out} \le \epsilon_s$. By definition \begin{eqnarray*}
P_{s, out} &=&  P\left(\frac{P_s d_s^{-\alpha}|\bq_{00}\bu_{0}|^2}{\sum_{n:\ T_{sn} \in \Phi_s \backslash \{T_{s0}\}}P_sd_{ss, n}^{-\alpha}|\bq_{0n}\bu_n|^2 + 
\sum_{n:\ T_{pn} \in \Phi_p }P_p d_{ps, n}^{-\alpha}|f_{0n}|^2} \le \beta_s\right), \\
&\stackrel{(a)}=&  P\left(\frac{P_s d_s^{-\alpha}|\bq_{00}\bu_0|^2}
{\sum_{n:\ T_{n} \in \Phi \backslash \{T_{s0}\}}d_{n}^{-\alpha}|\bq_{0n}\bu_0|^2 } \le \beta_s\right), \\
&\stackrel{(b)}\le & c_5 (\lambda_p P_p^{\frac{2}{\alpha}} + \lambda_s P_s^{\frac{2}{\alpha}})\left(\frac{\beta_s}{d_{c}(N-k)}\right)^{\frac{2}{\alpha}},\\
&\stackrel{(c)}\ge& c_6 (\lambda_p P_p^{\frac{2}{\alpha}} + \lambda_s P_s^{\frac{2}{\alpha}})\left(\frac{\beta_s}{d_{c}(N-k)}\right)^{\frac{2}{\alpha}}
\end{eqnarray*}
where $(a)$ follows from Lemma \ref{lem:sumintf} by letting $\Phi = \{\Phi_s \cup \Phi_p\}$ and Lemma \ref{lem:sisointdist}, since $|\bq_{0n}\bu_n|^2$ is Chi-square distributed with $2$ DOF similar to $|f_{0n}|^2$, and $(b)$ and $(c)$ follows from \cite{Hunter2008}, since $|\bq_{00}\bu_0|^2$ is a Chi-square random variable with $2(N-k)$ DOF Lemma \ref{lem:sisosigdist} for constants $c_5$ and $c_6$. Thus, considering the outage probability constraint of $\epsilon_s$ for the secondary network, 
 \begin{equation}\label{eq:lbubsecmiso}
 \lambda_s = \Theta\left((N-k)^{\frac{2}{\alpha}}\right).
 \end{equation}

Combining (\ref{eq:lbubmiso}) and (\ref{eq:lbubsecmiso}), $\lambda_s^{\star} = \Omega\left(\min\left\{k^{1-\frac{2}{\alpha}}, (N-k)^{\frac{2}{\alpha}}\right\}\right)$, and 
$\lambda_s^{\star} = {\cal O}\left(\min\left\{N,(N-k)^{\frac{2}{\alpha}}\right\}\right)$.
Hence $k=\theta N, \theta \in (0,1]$, provides the best scaling of the intensity of the secondary network and results in $\lambda_s^{\star} = \Omega\left(\min\{N^{1-\frac{2}{\alpha}}, N^{\frac{2}{\alpha}}\}\right)$, and $\lambda_s^{\star} =  {\cal O}\left((N-k)^{\frac{2}{\alpha}}\right)$.

\section{}\label{app:mimo}
Considering the outage probability at any primary receiver when each secondary transmitter uses $N-1$ DOF for nulling, from Theorem \ref{thm:Ntx1}
\begin{equation}
\label{eq:scalingmimoprim}
\lambda_s = \Omega\left( N^{1-\frac{2}{\alpha}}\right), \text{and} \ \ \lambda_s = {\cal O}\left( N\right).
\end{equation}
Next, we evaluate the maximum $\lambda_s$ that satisfies the outage probability constraint of $\epsilon_s$ for the secondary network with rate $R_p$ bits/sec/Hz for each transmission.  Using Lemma \ref{lem:2repcogintf}, the outage probability for the secondary pair $T_{s0} R_{s0}$ is 
\begin{eqnarray*}
P_{s, out} &\stackrel{(a)}=&  P\left(\frac{P_s d_s^{-\alpha}|\bt_0^*\bQ_{00}\bu_0|^2 }{\sum_{n>m, \ T_{n} \in \Phi \backslash \{T_{0}\}}P_nd_{n}^{-\alpha} I^{0n}} \le \beta_s\right), \\
&\stackrel{(b)}\ge&  1- \frac{(M-m)(m+1+c_4)^{\frac{2}{\alpha}}}{\frac{d_p^{\alpha}\beta_s }{P_s}(\pi (\lambda_s+\lambda_p))^{\frac{\alpha}{2}}}\left(\frac{\lambda_p}{P_p(\lambda_p+\lambda_s)} + \frac{\lambda_s}{P_s(\lambda_p+\lambda_s)}\right)\\
&\stackrel{(c)}\le & \frac{(\pi (\lambda_p+\lambda_s))^{\frac{\alpha}{2}}\beta_p d_p^{\alpha}\left(\left(\frac{\alpha}{2}-1\right)^{-1}(m+1)^{1-\frac{\alpha}{2}} + c_3\right)}{M-m-1}\left(\frac{\lambda_pP_p}{\lambda_p+\lambda_s} + \frac{\lambda_sP_s}{\lambda_p+\lambda_s}\right),
\end{eqnarray*}
where in $(a)$ $\Phi= \{\Phi_s \cup \Phi_p\}$ Lemma \ref{lem:2repcogintf}, and $I^{0n}$ is a Chi-square distributed random variable with $2$ DOF (Lemma \ref{lem:sigpowmimo}), and $(b)$, and $(c)$  follow from \cite{Vaze2009} and \cite{Jindal2008a}, respectively, with $|\bt_0^*\bQ_{00}\bu_0|^2$ Chi-square distributed  with $2(M-m)$ DOF (Lemma \ref{lem:sigpowmimo}),  and after taking the expectation with respect to  $P_n$.
Thus, using $m = \theta M, \ \theta \in (0,1],$ provides the best scaling of the intensity and results in
\begin{equation}
\label{eq:scalingmimocog}
\lambda_s = \Omega (M), \text{and} \ \ \lambda_s = {\cal O}(M^{1+\frac{2}{\alpha}}). 
\end{equation}
Hence combining (\ref{eq:scalingmimoprim}), and (\ref{eq:scalingmimocog}), $\lambda_s^{\star} =  \Omega\left(\min\{M, N^{1-\frac{2}{\alpha}} \}\right) $, and $\lambda_s^{\star} = {\cal O}\left(\min\left\{N, M^{1+\frac{2}{\alpha}}\right\}\right)$.

\bibliographystyle{IEEEtran}
\bibliography{IEEEabrv,Research}

\begin{figure}
\centering
\includegraphics[width=3.5in]{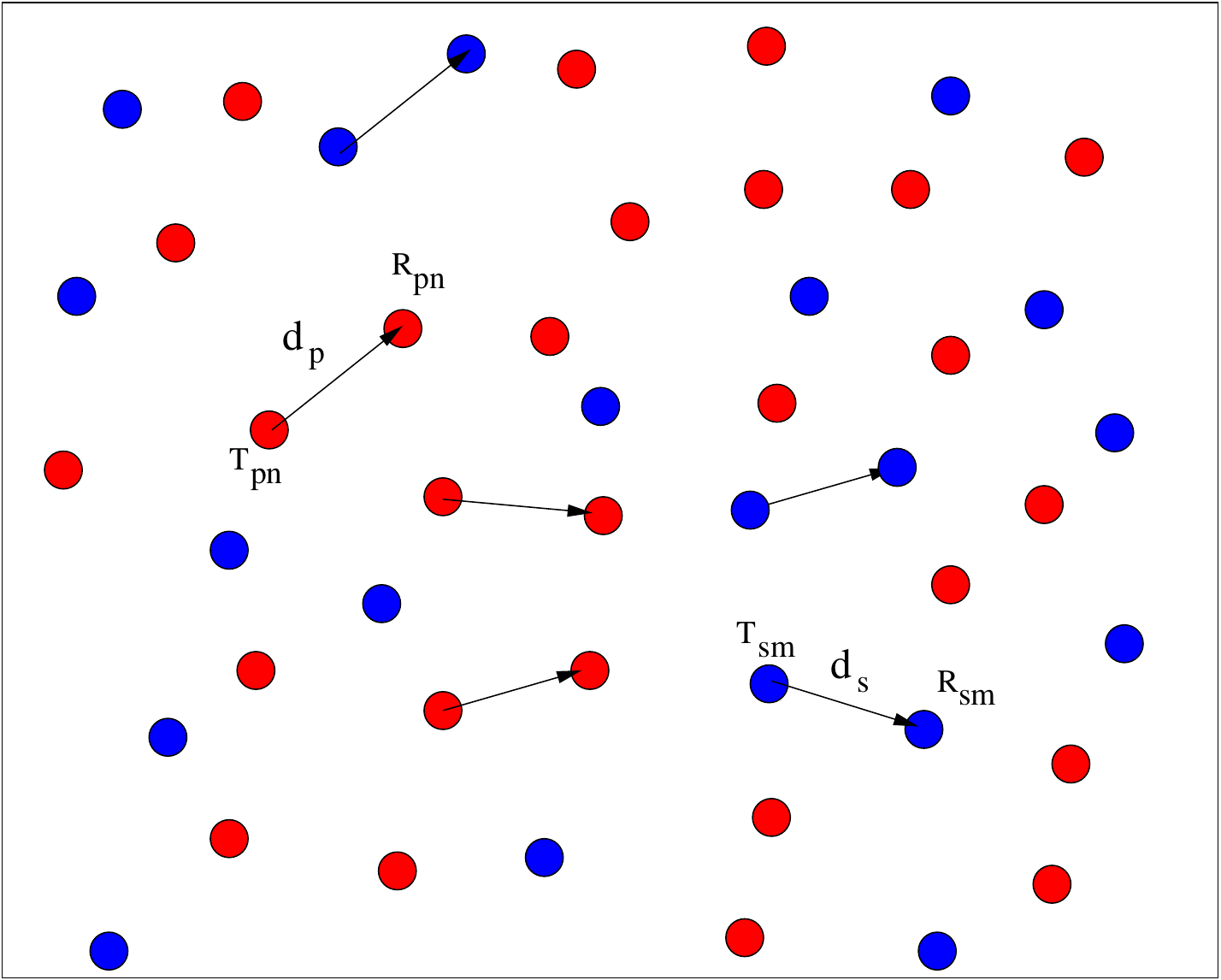}
\caption{Schematic of the coexisting ad hoc networks.}
\label{fig:blkdiag}
\end{figure}

\begin{figure}
\centering
\includegraphics[width=4.5in]{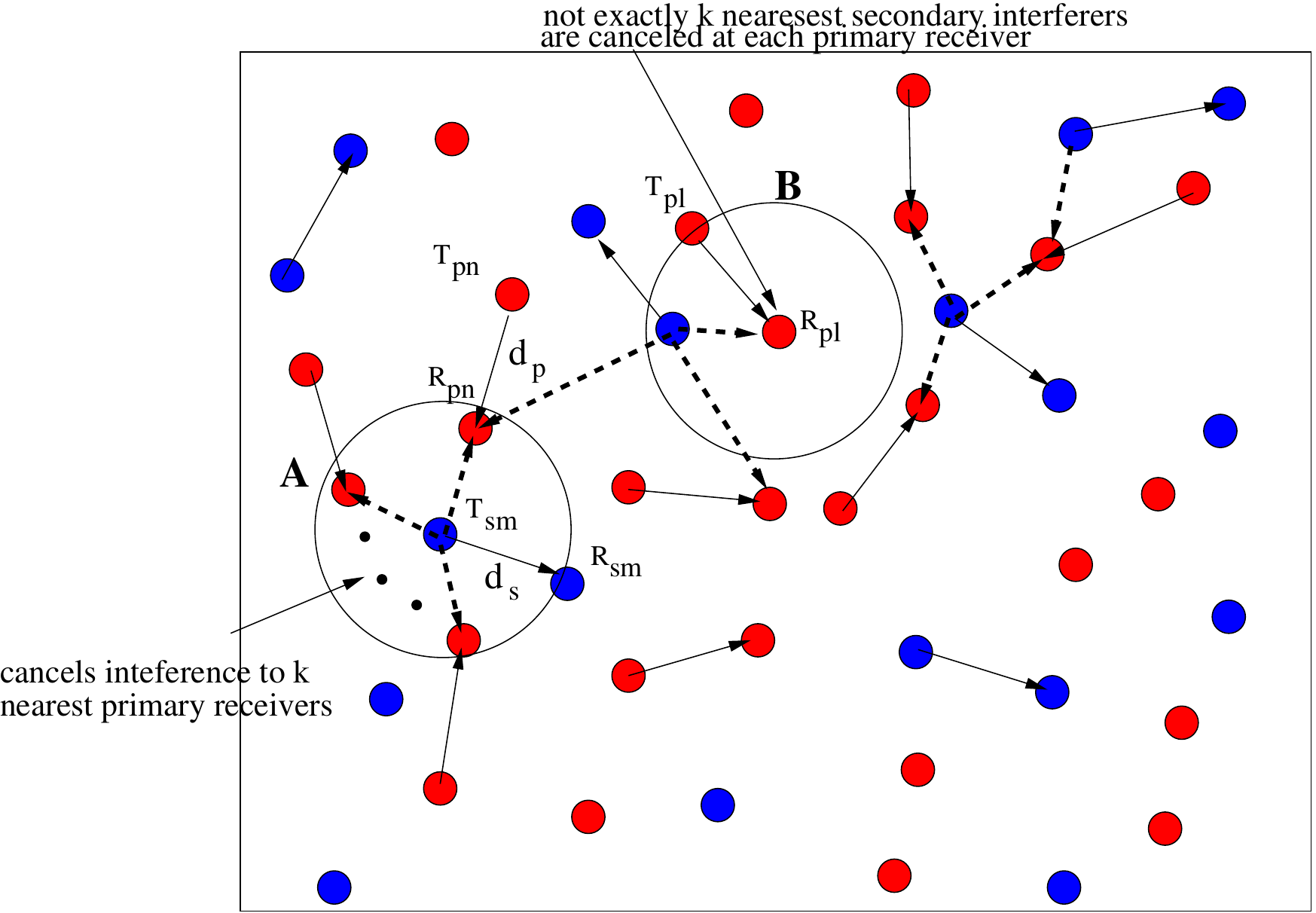}
\caption{Illustration of interference nulling by secondary transmitters towards their nearest primary receivers.}
\label{fig:nearcanc}
\end{figure}

\begin{figure}
\centering
\includegraphics[width=6.5in]{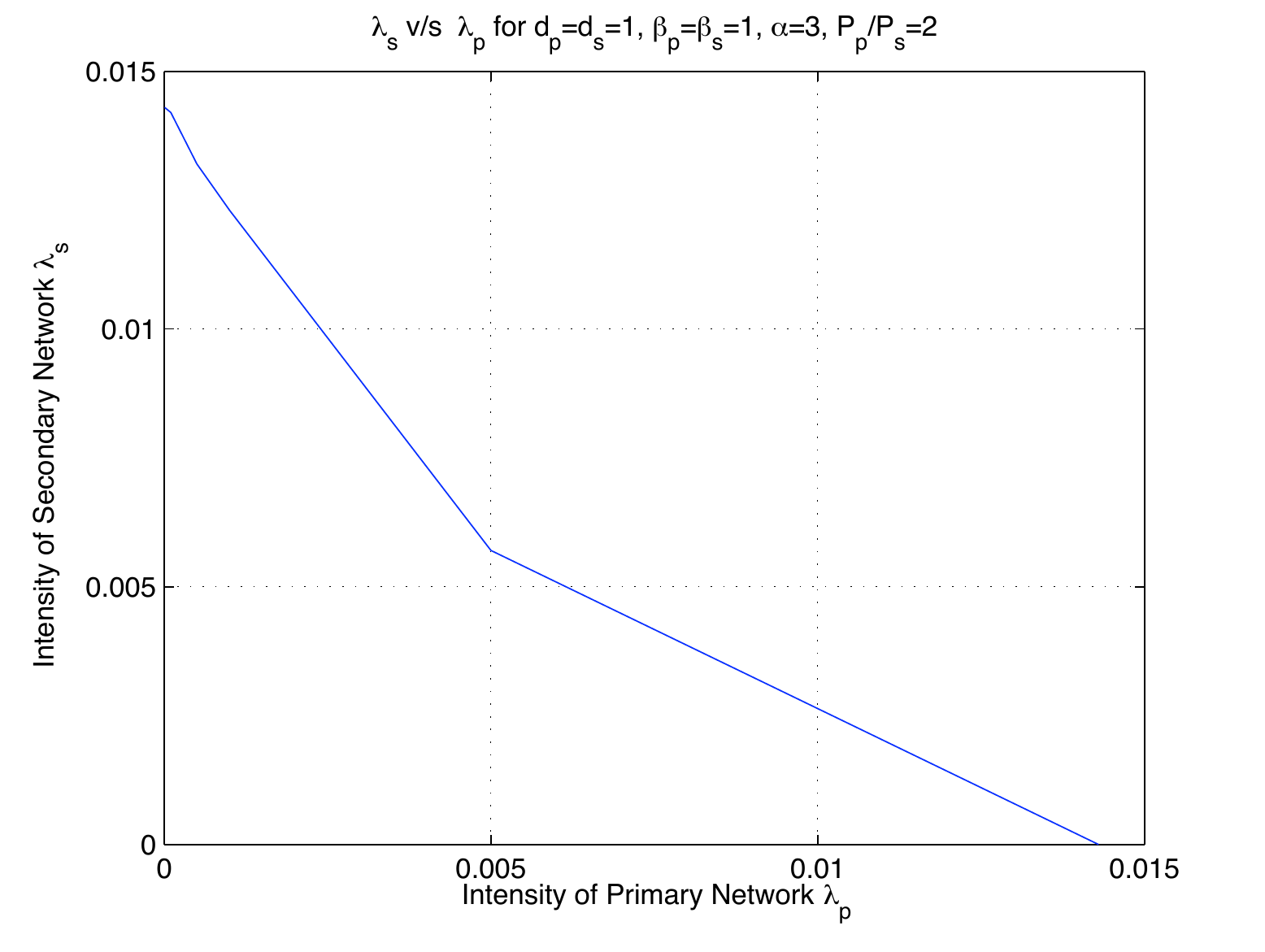}
\caption{Transmission capacity of the secondary network v/s transmission capacity of the primary  network.}
\label{fig:svp}
\end{figure}

\begin{figure}
\centering
\includegraphics[width=6.5in]{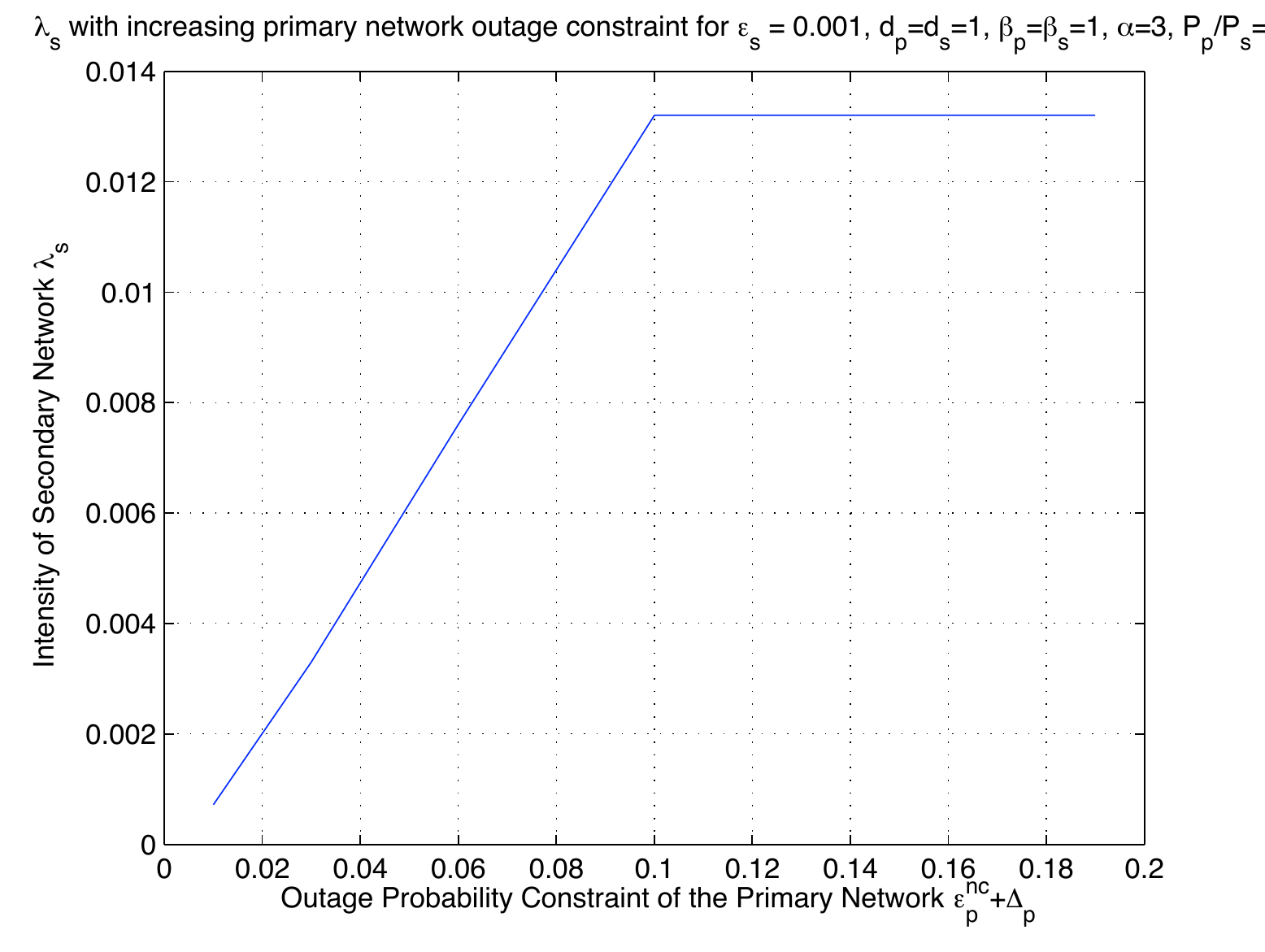}
\caption{Transmission capacity of the secondary network with increasing primary outage probability constraint.}
\label{fig:svep}
\end{figure}

\begin{figure}
\centering
\includegraphics[width=4.5in]{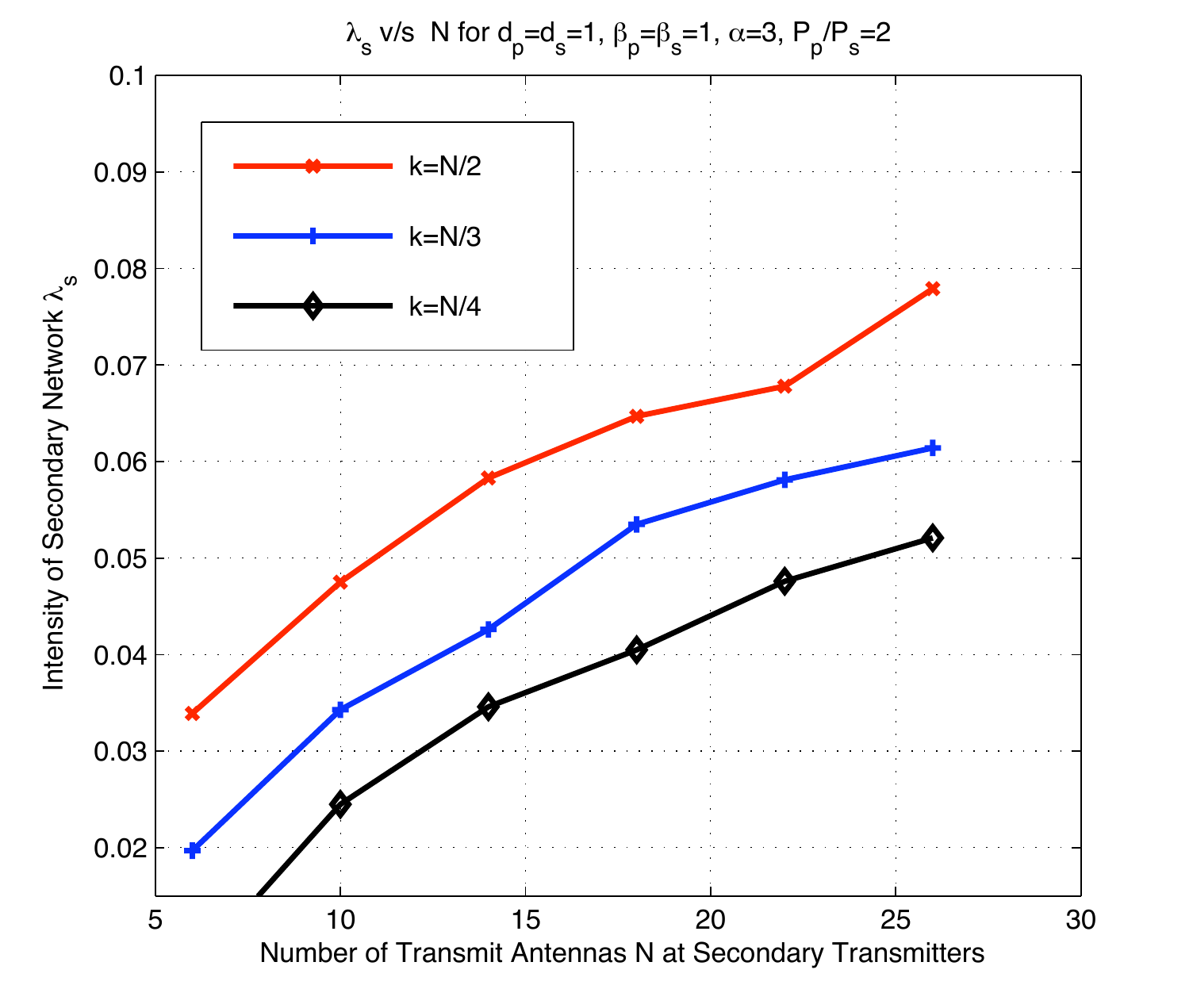}
\caption{Transmission capacity of the secondary network with multiple transmit antennas and $M=1$.}
\label{fig:smiso}
\end{figure}

\begin{figure}
\centering
\includegraphics[width=4.5in]{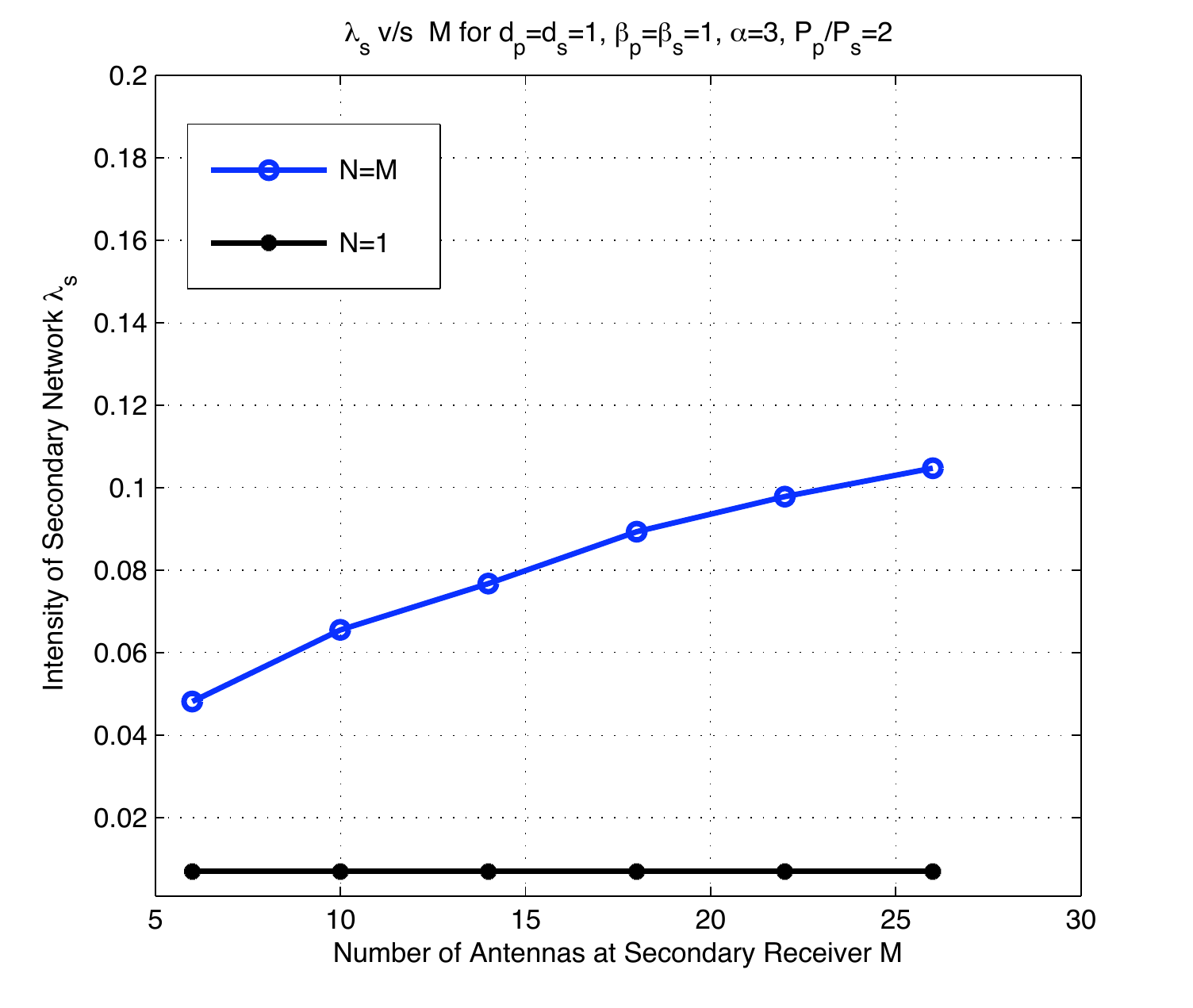}
\caption{Transmission capacity of the secondary network with multiple transmit and receive antennas.}
\label{fig:smimo}
\end{figure}

\end{document}